\documentclass{article}
\usepackage[a4paper, total={6in, 8in}]{geometry}
\usepackage{caption}
\usepackage{subcaption}
\usepackage{amsthm}
\newtheorem{theorem}{Theorem}

\usepackage{amssymb}
\usepackage{lscape}
\usepackage{authblk}
\usepackage{hyperref}
\usepackage{rotating}
\usepackage[square,sort,comma,numbers]{natbib}

\usepackage{color}

\title{More can be better: An analysis of single-mutant fixation probability functions under $2\times2$ games
}
\author{Diogo L. Pires\footnote{\href{mailto:Diogo.L.Pires@city.ac.uk}{Diogo.L.Pires@city.ac.uk}} }
\author{Mark Broom\footnote{\href{mailto:Mark.Broom.1@city.ac.uk}{Mark.Broom.1@city.ac.uk}} }
\affil{\small{City, University of London, Northampton Square, London, EC1V 0HB, UK}}
\date{\today}

\usepackage{mathtools}
\usepackage{graphicx}
\usepackage{float}
\usepackage{multirow}
\usepackage{appendix}
\usepackage{multicol}

\begin{document}
\maketitle

\begin{abstract}
Evolutionary game theory has proved to be a powerful tool to probe the self-organisation of collective behaviour by considering frequency-dependent fitness in evolutionary processes. It has shown that the stability of a strategy depends not only on the payoffs received after each encounter but also on the population's size. Here, we study $2\times2$ games in well-mixed finite populations by analysing the fixation probabilities of single mutants as functions of population size. We proved that 9 out of the 24 possible games always lead to monotonically decreasing functions, similarly to fixed fitness scenarios. However, fixation functions showed increasing regions under 12 distinct anti-coordination, coordination, and dominance games. Perhaps counter-intuitively, this establishes that single-mutant strategies often benefit from being in larger populations. Fixation functions that increase from a global minimum to a positive asymptotic value are pervasive but may have been easily concealed by the weak selection limit. We obtained sufficient conditions to observe fixation increasing for small populations and three distinct ways this can occur. Finally, we describe fixation functions with increasing regions bounded by two extremes under intermediate population sizes. We associate their occurrence with transitions from having one global extreme to other shapes.
\end{abstract}

\section{Introduction}

Evolutionary game theory has proved to be an essential framework in the context of studying the self-organisation of collective behaviour such as signalling \cite{MaynardSmith1991Signalling,Skyrms2010Signalling,Hutteger2014Signalling}, cooperation \cite{Axelrod1984,Poundstone1992,Skyrms2004}, or aggressive territorial behaviour \cite{MaynardSmith1973,BroomRychtar2013}. The fitness of individuals in a population was first considered to depend on the frequency of their types in evolutionary models of infinite populations, under which either a static analysis of Evolutionarily Stable Strategies (ESS) \cite{MaynardSmith1973} or finding attractors of the replicator equation \cite{Taylor1978,Hofbauer1998} indicated stable solutions of the evolutionary dynamics.

Only later in \cite{Nowak2004}, the original Moran process \cite{Moran1958} was adapted to incorporate frequency-dependent fitness, thus initiating the study of evolutionary games on finite populations. Soon afterwards, it was shown that the stability of strategies depends not only on the payoffs received by individuals but on the size of the considered finite population as well \cite{Taylor2004}. For strong selection, however, the evolutionary outcomes were shown to correspond to the classical infinite population ones, first under the frequency-dependent Moran process \cite{DellaRossa2017Unifying}, and then under general fitness mappings \cite{Huang2018}.

These models introduced a new concept of evolutionary stability $ESS_N$ adapted from the definition used under infinite populations \cite{MaynardSmith1973}. An evolutionarily stable strategy became one on which selection opposes both invasion and fixation by other strategies \cite{Schaffer1988ESSN,Nowak2004}. These two requirements respectively provide equilibrium and stability conditions \cite{BroomRychtar2013} of a given strategy. Selection opposes invasion when no single mutant holds a higher fitness than the residents in an otherwise pure population. Selection opposes fixation when no single mutant is able to fixate in the population with a probability above neutral fixation \cite{Nowak2004}. It was shown that of the $16$ possible combinations of selection favouring or opposing the invasion and fixation of strategists $B$ by strategists $A$ and vice-versa, only $8$ of those are possible evolutionary scenarios under the frequency-dependent Moran process \cite{Taylor2004}.

Studying fixation probabilities thus became an important aspect of understanding the evolution of collective behaviour. They were studied as functions of the initial number of mutants \cite{deSouza2019plotshape}, and it has been proved that each of the $8$ possible evolutionary scenarios has associated to it one (and only one) of three function shapes. In  \cite{Traulsen2006,Traulsen2007}, they were studied in relation to the initial number of mutants and intensity of selection under the pairwise comparison process. They have also been extensively studied in the context of structured populations \cite{BroomHadji2010Curve,Hadji2011StarDynamics,Hindersin2016FixationGraphs,Allen2021FixationGraphWeakSelection}.

Fixation probabilities describe the likelihood of changing from one pure state to another after one mutation occurs, conditional on mutations being rare enough. Following \cite{Fudenberg2006SML}, under games with two or more strategies they can be used to construct a simplified Markov chain between pure states, and allow the computation of the long-term stationary distribution over them \cite{Hauert2007SMLSimulations,SegbroeckSantos2009Diversity,Santos2011PrePlaySignalling,Wagner2020Signalling}. The validity of these approximations under rare mutations has been analysed analytically in \cite{Wu2012SML} and through comparison with simulations in \cite{Hauert2007SMLSimulations,SegbroeckSantos2009Diversity}. The concerns raised over the existence of mixed stable states leading to quasi-stationary distributions \cite{Nasell1999QuasiStationary1,Nasell1999QuasiStationary2,Zhou2010} have been tackled in this context by considering higher-orders approximations including those states as configurations of interest \cite{Vasconcelos2017Hierarchic}.

The dependence of fixation probabilities on population size has been analysed, for instance through the calculation of their asymptotic limits. This was initially done by \cite{AntalScheuring2006}, who obtained quantitative limits for some cases, but only qualitative ones in others, whose error was later estimated by \cite{deSouza2019plotshape}. These results were later expanded for two different limit orders of weak selection by \cite{SampleAllen2017}, who also corrected some borderline cases under arbitrary values of intensity of selection.

Taking into consideration the central importance that single-mutant fixation probabilities have in the context of evolutionary theory, the present work analyses them systematically as functions of population size under the simplest $2\times 2$ games. To do so, we classify games based on the ordering of their payoff matrix' entries, leading to a total of $24$ different fixation orderings. We start by proving in section \ref{sec:fixprob_decreasing} that nine of these always lead to monotonically decreasing fixation probability functions. We tested the remaining orderings and concluded that three additional orderings may not have any non-monotonically decreasing functions. In total, these included the fixation of $6$ dominated strategies, one dominating strategy, and five strategies from coordination games.

However, increasing population size may not only change whether mutants fixate below or above neutrality \cite{Nowak2004,Taylor2004}, but can correspond to actual increases of single-mutant fixation probabilities. We observed diverse fixation functions with increasing regions under the twelve remaining orderings. These included all six strategies from anti-coordination games such as the Hawk-Dove/Snowdrift game \cite{MaynardSmith1973,BroomRychtar2013,HauertDoebeli2004,DoebeliHauert2005}, the fixation of five dominating strategies such as defectors in the Prisoner’s Dilemma \cite{Axelrod1984,Poundstone1992}, and the fixation of stag hunters in the Stag Hunt game \cite{Skyrms2001,Skyrms2004} (the only exception in coordination games). In section \ref{sec:fixprob_min}, fixation functions that increased from a global minimum to a positive asymptotic value were explored and found to be pervasive. In some anti-coordination games (e.g. fixation of doves) this shape was found every time the payoff matrix led to a positive asymptotic value. These functions have passed mostly unnoticed in the past with the exception of \cite{BroomHadji2010Curve}, where it was briefly noted that they could be observed under the Hawk-Dove game. We propose that they may have been hidden by the weak selection limit \cite{Nowak2004,Traulsen2006,Traulsen2007,Wild2007WeakSelection}, especially if this limit was considered to be dominant over the large population one \cite{SampleAllen2017}. 

In section \ref{sec:fixprob_init_increasing}, we show that it should be possible to see fixation increasing for the smallest populations $N=2$ under $6$ different orderings: three dominating, two anti-coordination, and one coordination game strategy orderings. We find three different ways this can happen: functions might increase monotonically, or they increase up to a global maximum and then decrease to a positive asymptotic value or to zero. 

Finally, in section \ref{sec:fixprob_two_extremes}, we explore fixation functions with two extremes: decreasing for small populations, increasing for intermediate populations, and decreasing again for larger populations. These were observed both with positive and zero asymptotic values. The first were observed under the fixation of two dominating strategies, and the second under anti-coordination games and the Stag Hunt game. These were all observed for transitions between functions with one global extreme and other function shapes.

\section{Model}

Let us consider a well-mixed population of individuals who interact pairwise according to a general $2\times 2$ symmetric game, and call $A$ and $B$ the strategies at their disposal. After each encounter, they receive a payoff defined by the payoff matrix from table \ref{tab:2x2_payoff_matrix}. Individuals using $A$ receive $a$ and $b$ respectively against individuals using $A$ and $B$; while individuals using $B$ receive $c$ and $d$ respectively against individuals using $A$ and $B$. We consider all payoff matrix entries to be strictly positive.

\begin{table}[h]
\centering
\begin{tabular}{c|cc}
    & $A$ & $B$ \\ \hline
$A$ & a   & b   \\
$B$ & c   & d  
\end{tabular}
\caption{Payoff matrix of a $2\times2$ game.}
\label{tab:2x2_payoff_matrix}
\end{table}

Each individual in a well-mixed population interacts on average with the same frequency with all others. In this context, their fitness is simply the average of the payoffs received over the encounters they have with other individuals in the population. Considering a population with $N$ individuals, where $i$ individuals are using strategy $A$ and $N-i$ using $B$, the fitness of individuals using $A$ and $B$ are, respectively, the following:
\begin{equation}
\label{eq:fitness_A}
  f_i^N= \frac{a (i-1)+b (N-i)}{N-1},
\end{equation}
\begin{equation}
\label{eq:fitness_B}
  g_i^N = \frac{c i + d (N-i-1)}{N-1}.
\end{equation}

However, selection may depend on factors beyond the studied game, which on average should represent equal contributions to each individual's fitness, irrespective of their strategy. To account for this, we use the fitness formulation proposed in \cite{Nowak2004}, which includes a parameter representing the intensity of selection $w\in [0,1]$:
\begin{equation}
\label{eq:fitness_A_w}
  f_i^N= 1-w+w \frac{a (i-1)+b (N-i)}{N-1},
\end{equation}
\begin{equation}
\label{eq:fitness_B_w}
  g_i^N = 1-w+w \frac{c i + d (N-i-1)}{N-1}.
\end{equation}

Changing $w$ is equivalent to performing a transformation of the original payoff matrix. However, as this transformation does not change the original game ordering, we will set $w=1$ for most of our analysis (leading back to equations \ref{eq:fitness_A} and \ref{eq:fitness_B}), except when intending to study the explicit effects of intensity of selection.

Selection may act in different ways on the evolutionary process, and hence the fitness of each individual may be considered on the first (birth) event \cite{Nowak2004,Taylor2004}, on the second (death) event \cite{Ohtsuki2006ANetworks}, on simultaneous events \cite{Pattni2017Subpopulations}, or to depend exponentially on both individuals' payoffs \cite{Traulsen2006}. Here we are focusing on the birth-death process with selection in the birth event, which is typically described as the frequency-dependent Moran process \cite{Nowak2004,Taylor2004}. This is a process under which it has been suggested that cooperators perform generally worse \cite{Ohtsuki2006ANetworks,Pattni2015EvolutionaryProcess} when compared to the results obtained under the others.

During each step of this process, one individual in the population gives birth proportionally to their fitness and another one dies randomly in the population. Taking this into consideration, for each evolutionary step, the probabilities of having the number of individuals $i$ using strategy $A$ increasing by one $\left(P_{i^+}^N\right)$, decreasing by one $\left(P_{i^-}^N\right)$, or remaining the same $\left(P_{i^=}^N\right)$ are the following:
\begin{equation}
\label{eq:transition_plus}
P_{i^+}^N = \frac{if_i^N}{if_i^N+(N-i)g_i^N} \frac{N-i}{N},
\end{equation}
\begin{equation}
\label{eq:transition_minus}
P_{i^-}^N = \frac{(N-i)g_i^N}{if_i^N+(N-i)g_i^N} \frac{i}{N},
\end{equation}
\begin{equation}
\label{eq:transition_equal}
P_{i^=}^N = 1 - P_{i^+}^N - P_{i^-}^N.
\end{equation}

All the other transitions are impossible by definition of the dynamics. This stochastic process is defined as a Markov chain, under which there are two absorbing states: $i=0$ and $i=N$. If the population falls into one of these, it will stay there for the remaining evolutionary time.

To compute the probability of one single mutant $A$ fixating in the population, i.e. seeing the system transitioning from $i=1$ to $i=N$, we have to consider a recursive relation based on equations \ref{eq:transition_plus}, \ref{eq:transition_minus} and \ref{eq:transition_equal}. Following \cite{Karlin1975Stochastic}, this relation leads to a closed-form expression for the fixation probability
\begin{equation}
\label{eq:karlin_fixation_probability}
  \rho_N=\frac{1}{1+ \sum_{j=1}^{N-1} \prod_{k=1}^j \gamma_k^N},
\end{equation}
where $\gamma_k^N$ is defined as
\begin{equation}
\label{eq:transition_prob_ration}
\gamma_k^N=\frac{P_{k^-}^N}{P_{k^+}^N}.
\end{equation}

Applying equations \ref{eq:transition_plus} and \ref{eq:transition_minus}, the term $\gamma_k^N$ reduces to the relative fitness of residents in a population with $k$ mutants and $N$ total individuals. It will prove useful to express it as an explicit function of the number of mutants $k$ and the population size $N$ when $w=1$:

\begin{equation}
\label{eq:relative_fitness_fd_moran}
  \gamma_k^N = \frac{g_k^N}{f_k^N} = \frac{c k + d (N-k-1)}{a (k-1) + b (N-k)}.
\end{equation}

\pagebreak

\section{Results}

General $2\times 2$ matrix games are often defined by the ordering of the four payoff matrix entries from table \ref{tab:2x2_payoff_matrix} (see e.g. \cite{BroomRychtar2013,Traulsen2007}). In the present work, we have only considered cases where all payoff matrix entries are positive and distinct pairwise. Even though there are $24$ possible orderings, the pairs that have $b$ and $c$ swapped as well as $a$ and $d$ correspond to the same game with the definitions of strategies $A$ and $B$ swapped. Therefore, there are a total of $12$ independent $2\times2$ games, each with a pair of complementary orderings. These pairs are shown together in table \ref{tab:orderings} and delimited by horizontal lines. In each ordering, strategy $A$ is considered to be the mutant and strategy $B$ the resident.

In infinite populations, equilibria under these games are limited to three scenarios \cite{MaynardSmith1973,Taylor1978,Hofbauer1998,Traulsen2009Stochastic}. Games either have one pure evolutionarily stable state ($ESS$), one inner $ESS$, or two pure $ESS$. These are respectively named games of dominance, anti-coordination, and coordination, and  directly correspond to three invasion scenarios in stochastic models when $N\to\infty$ \cite{Taylor2004}. The correspondence between these games and the ordering of their payoffs is set in the first two columns of table \ref{tab:orderings}. For each pair of dominance orderings on the table, the first ordering corresponds to having mutant $A$ invading $B$, while the second corresponds to $B$ invading $A$. All other pairs of complementary game orderings did not follow any particular order.

%\begin{landscape}
\begin{sidewaystable}[p]
\resizebox{\columnwidth}{!}{%
\begin{tabular}{c|ccccccc}
\hline
\hline
%\multirow{2}{*}{\begin{tabular}[c]{@{}c@{}}Type of Game under $N\to\infty$\end{tabular}}
Type of game
& \multirow{2}{*}{\begin{tabular}[c]{@{}c@{}}Game\end{tabular}}
%\begin{tabular}[c]{@{}c@{}}Game\end{tabular} 
& \multirow{2}{*}{\begin{tabular}[c]{@{}c@{}}Ordering\end{tabular}}
& \multicolumn{2}{c}{Shapes of fixation functions $\rho_N$ with}
& \multirow{2}{*}{\begin{tabular}[c]{@{}c@{}}Always decreasing $\rho_N$\end{tabular}}
& \multirow{2}{*}{\begin{tabular}[c]{@{}c@{}}Possibility of $\rho_3>\rho_2$\end{tabular}} 
&  \\
under $N\to\infty$ &&&positive asymptotic value & null asymptotic value &&&\\
\hline
\multirow{12}{*}{\begin{tabular}[c]{@{}c@{}}Dominance Games\\
(one pure $ESS$) \end{tabular}}

& PD (D) & $b>d>a>c$    &DP DUP             &      &  &  & \\  
& PD (C) & $c>a>d>b$    &                   & D0 & tested    &  &  \\ \cline{2-8} 
&   & $a>c>b>d$         & DP DUP UP UDP DUDP&      &   &    yes   &  \\
&   & $d>b>c>a$         &                   & D0 & tested  &  & \\ \cline{2-8} 
&   & $a>b>c>d$         & DP DUP UP UDP DUDP &      &  &   yes     &  \\
&   & $d>c>b>a$         &                   & D0   & proven (T1+T2)  &   &  \\ \cline{2-8} 
&   & $a>b>d>c$         & DP                &      & proven (T1) &   &  \\
&   & $d>c>a>b$         &                   & D0   & proven (T1) &   &  \\ \cline{2-8} 
&  & $b>a>d>c$   & DP DUP &  &  &   &  \\
&  & $c>d>a>b$  &  & D0 & tested &   &  \\ \cline{2-8} 
&   & $b>a>c>d$ & DP DUP UP &  &   &    yes  &  \\
&   & $c>d>b>a$  & & D0  & proven (T2) &   &  \\ \hline
\multirow{6}{*}{\begin{tabular}[c]{@{}c@{}}Anti-Coordination Games\\
(one mixed $ESS$) \end{tabular}}
& HD (D) & $c>a>b>d$ & DUP & D0 DUD0 &  &  &  \\
& HD (H) & $b>d>c>a$ & DP DUP & D0 DUD0 & &  &  \\ \cline{2-8} 
&   & $b>c>a>d$ &  DUP UP &&   &   yes  &  \\
&   & $c>b>d>a$ & & D0 DUD0 &   &  &  \\ \cline{2-8} 
&   & $b>c>d>a$ &  DUP UP & D0 DUD0 &   &   yes &  \\
&   & $c>b>a>d$ &  DUP & D0 DUD0 &   &  &  \\ \hline
\multirow{6}{*}{\begin{tabular}[c]{@{}c@{}}Coordination Games\\
(two pure $ESS$)\end{tabular}} 
& SH (S)   & $a>c>d>b$ & & D0 UD0 DUD0 &  &   yes &  \\
& SH (H)   & $d>b>a>c$ & & D0 & proven (T1) &   &  \\ \cline{2-8} 
&  & $a>d>b>c$ & & D0 & proven (T1) &   &  \\
&   & $d>a>c>b$ & & D0 & proven (T1)  &   &  \\ \cline{2-8} 
&   & $a>d>c>b$ & & D0 & proven (T1)  &   &  \\
&   & $d>a>b>c$ & & D0  & proven (T1)  &   &  \\ \hline
\hline
\end{tabular}}
\caption{Summary of fixation probability results under general $2\times 2$ games. Games can be defined by the ordering of the payoff matrix's entries. There is a total of $24$ orderings, corresponding to $12$ independent $2\times2$ games with distinguishable strategies. Pairs of orderings of the same game are enclosed between horizontal lines. We split games into three types -- dominance, anti-coordination and coordination games -- depending on the equilibria observed in the asymptotic limit $N\to\infty$ under them. Orderings corresponding to known games are signalled -- Prisoner's Dilemma (PD), Hawk-Dove or Snowdrift (HD), Stag Hunt (SH) -- and the strategy considered to be mutant $A$ is indicated in parenthesis. The first/second ordering of each pair of dominance games corresponds to the fixation of the dominating/dominated strategy. Under the other two types of games, pairs did not follow any particular ordering. The shapes of fixation functions $\rho_N$ found under each of the orderings are listed, separated by the ones having positive and null asymptotic values. The last two columns summarise the results of theorems from sections \ref{sec:fixprob_decreasing} and \ref{sec:fixprob_init_increasing} respectively. Some game orderings were \textit{proven} to always have decreasing fixation functions $\rho_N$, either by Theorem $1$ (T1) or by Theorem $2$ (T2). One of the orderings resulted from a joint proof (T1+T2) for its two subsets: $bc\leq ad$ by T1 and $bc\geq ad$ by T2. We did not include the results for incomplete subsets of other orderings. Some orderings were observed to always lead to decreasing functions after being \textit{tested} systematically. The remaining orderings were found to have non-decreasing functions, as can be confirmed by the listing of shapes found under each ordering. Condition $\rho_3>\rho_2$ was met under the orderings signalled in the last column when eq. \ref{eq:condition_increasing2_4} was fulfilled.}
\label{tab:orderings}
\end{sidewaystable}
%\end{landscape}

However, under finite populations, evolutionary outcomes are defined not only by whether selection favours or opposes invasion, but also by whether it does so with fixation  \cite{Taylor2004}. In the following sections, we will explore some of the possible ways in which fixation probabilities may depend on population size in $2\times2$ games. We found that there are at least $8$ shapes that single-mutant fixation probability functions $\rho_N$ can take:

\begin{enumerate}
    \item Always decreasing to a positive value (DP);
    \item Always decreasing to $0$ (D0);
    \item Decreasing to a global minimum and then increasing up to a positive value (DUP); 
    \item Always increasing up to a positive value (UP); 
    \item Increasing up to a global maximum and then decreasing to a positive value (UDP);
    \item Increasing up to a global maximum and then decreasing to $0$ (UD0); 
    \item Decreasing to a minimum, then increasing up to a maximum, and finally decreasing to a positive value (DUDP);
    \item Decreasing to a minimum, then increasing up to a maximum, and finally decreasing to $0$ (DUD0).
\end{enumerate}

We observed a higher diversity of fixation functions with increasing regions mainly under anti-coordination games (e.g. Hawk-Dove/Snowdrift game), the fixation of dominating strategies (e.g. defectors in the Prisoner’s Dilemma), and the fixation of stag hunters under the game with the same name (the only exception in coordination games). A summary of the fixation function shapes can be found in figure \ref{fig:fix_prob}. In table \ref{tab:orderings}, we have represented both the fixation function shapes observed under each game, and how the analytical results reflected on them. In section \ref{sec:fixprob_decreasing}, we prove that some orderings always have decreasing fixation probability functions, and then explore these functions (shapes 1 and 2). In section \ref{sec:fixprob_min}, we explore functions with one global minimum (shape 3), state these are pervasive across dominance and anti-coordination games and suggest that these might have been concealed by the weak selection limit. In section \ref{sec:fixprob_init_increasing}, we obtain the conditions under which it is possible to see fixation probability functions increasing for very small populations and explore the functions under which this happens (shapes 4, 5 and 6). Finally, in section \ref{sec:fixprob_two_extremes}, we study the settings under which fixation probabilities are seen to have two extremes -- a local minimum and a local maximum -- before decreasing either to a positive value (shape 7) or to zero (shape 8) and associate these with transitions from functions with one global extreme (minimum or maximum) to other shapes.

\begin{figure}
\centering

\begin{subfigure}[b]{0.49\textwidth}
\centering
\includegraphics[width=\textwidth]{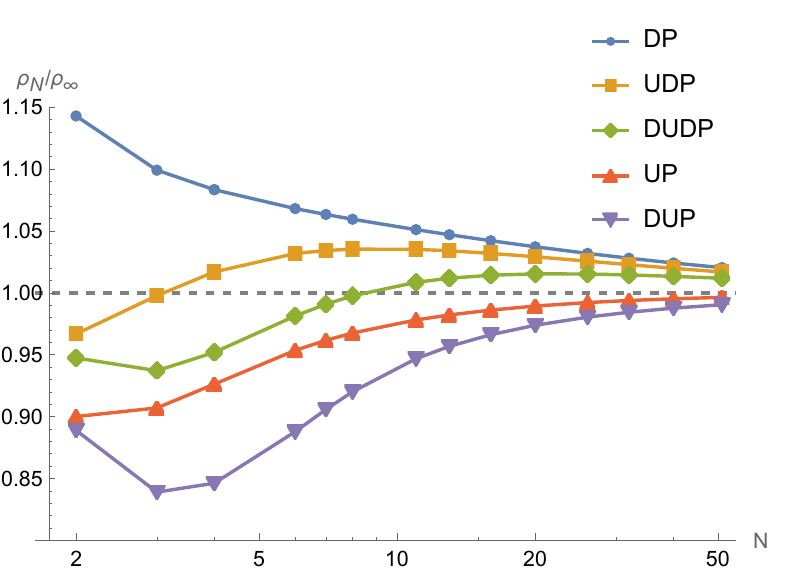}
\caption{\label{fig:fix_prob_finite_2}}
\end{subfigure}
%\vspace{0.4cm}
\begin{subfigure}[b]{0.49\textwidth}
\centering
\includegraphics[width=\textwidth]{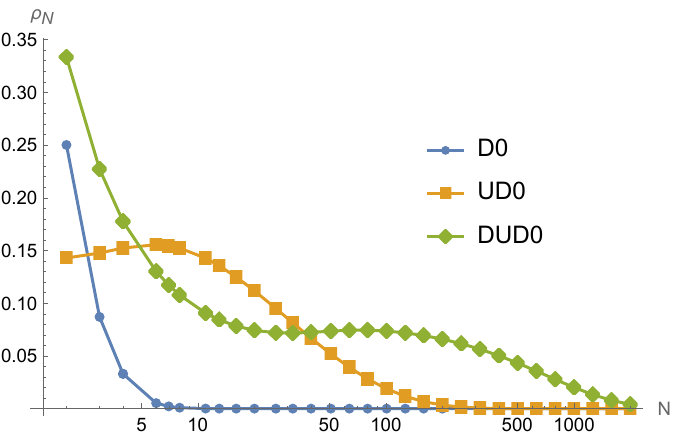}
\caption{\label{fig:fix_prob_null}}
\end{subfigure}
\caption[Summary of shapes of fixation probability functions under $2\times2$ games.]{\label{fig:fix_prob}Summary of shapes of single-mutant fixation probability functions $\rho_N$ under $2\times2$ games. Figure \ref{fig:fix_prob_finite_2} represents curves normalized by their positive asymptotic values ($\rho_{\infty}=1-d/b$) for simplification, and figure \ref{fig:fix_prob_null} represents curves with zero asymptotic values. The coding of function shapes describes their successive trends of going up (U) and down (D), and whether their asymptotic limit is positive (P) or zero ($0$). Results were obtained for the following payoff values: DP $[80,6,1,1.5]$, UDP $[20,3,1.9,1.1]$, DUDP $[8,1.8,2,0.9]$, UP $[10,3,2,1]$, DUP $[3,4,2,1]$, D0 $[2,1,3,4]$, UD0 $[13,0.5,3,1]$, DUD0 $[3,2,4,1.235]$. These payoff values were chosen to get optimal clarity of figures.}
\end{figure}

\subsection{\label{sec:fixprob_decreasing}Decreasing fixation probability functions}

Under the fixed fitness Moran process \cite{Moran1958}, increasing the size $N$ of a finite population always leads to a decrease in the probability that a single mutant has of fixating on the whole population. When frequency-dependent fitness is introduced \cite{Nowak2004}, despite this not being necessarily true, it is still observed to happen often. To probe the $2\times 2$ games under which $\rho_{N+1}>\rho_N$ is always true, we  obtain Theorems 1 and 2.

\begin{theorem}
If the payoff matrix entries of a $2\times 2$ game satisfy $b<a$ and $c\leq d$ or $b \geq a$ and $bc\leq ad$, 
then the fixation probability $\rho_N$ of a single mutant using $A$ decreases monotonically with $N$.
\end{theorem}

\begin{proof}

We try to prove that  $\rho_N>\rho_{N+1}$ is always true in some games. We start by comparing the closed-form expressions for the fixation probabilities recalling equation \ref{eq:karlin_fixation_probability} and focusing on their denominators:
\begin{equation}
\label{eq:condition_non_increasing}
\begin{split}
  \rho_N>\rho_{N+1} 
  &\Leftrightarrow \frac{1}{1+ \sum_{j=1}^{N-1} \prod_{k=1}^j \gamma_k^N}> \frac{1}{1+ \sum_{j=1}^{N} \prod_{k=1}^j \gamma_k^{N+1}} \Leftrightarrow \\
  &\\
  &\Leftrightarrow  \sum_{j=1}^{N-1} \prod_{k=1}^j \gamma_k^N < \sum_{j=1}^{N} \prod_{k=1}^j \gamma_k^{N+1}.\\
\end{split}
\end{equation}

Isolating the extra term on the sum on the right hand-side and joining the remaining two sums of $N-1$ terms, we get the following condition:
\begin{equation}
\label{eq:condition_non_increasing_1_1}
  \rho_N>\rho_{N+1} 
  \Leftrightarrow
  \sum_{j=1}^{N-1} \left[ \prod_{k=1}^j \gamma_k^N - \prod_{k=1}^j \gamma_k^{N+1} \right] <
  \prod_{k=1}^N \gamma_k^{N+1}.
\end{equation}

The right hand-side product is positive because all values of $\gamma_k^{N+1}$ are positive as well. Then, if the following more strict condition is fulfilled
\begin{equation}
\label{eq:condition_non_increasing_1_2}
  \sum_{j=1}^{N-1} \left[ \prod_{k=1}^j \gamma_k^N - \prod_{k=1}^j \gamma_k^{N+1} \right] \leq 0,
\end{equation}
we will necessarily have $\rho_N>\rho_{N+1}$.

Any payoff matrix fulfilling $\gamma_k^N\leq \gamma_k^{N+1} \forall k\in\{1,...,N-1\}; \forall N\in\{2,3,...\}$, satisfies condition \ref{eq:condition_non_increasing_1_2}. Thus, single-mutant fixation probabilities in such games will necessarily be decreasing functions of $N$. Under the frequency-dependent Moran Process, $\gamma_k^N$ is equivalent to the relative fitness of resident strategy $B$ as it is seen in equation \ref{eq:relative_fitness_fd_moran}. Therefore, $\gamma_k^N\leq \gamma_k^{N+1}$ leads to the following condition in that case:
\begin{equation}
\label{eq:condition_non_increasing_1_3}
\frac{c k + d (N-k-1)}{a (k-1) + b (N-k)} \leq \frac{c k + d (N-k)}{a (k-1) + b (N-k+1)}, 
\end{equation}
which after some algebra becomes
\begin{equation}
\label{eq:condition_non_increasing_1_4}
(bc-ad)k\leq (b-a)d.
\end{equation}

We are interested in the games that satisfy this condition for all values of $k\in\{1,...,N-1\}$ and $N\in\{2,3,...\}$, which is equivalent to saying that
\begin{equation}
\label{eq:condition_non_increasing_1_5}
bc-ad\leq \inf\left( \left\{ \frac{(b-a)d}{k}:k\in\{1,...,N-1\};N\in\{2,3,...\} \right\} \right).\\
\end{equation}

The value of $k$ for which the infimum occurs depends only on whether $(b-a)d$ is negative or positive.

\begin{enumerate}

\item Case $b<a$:

This case leads to $(b-a)d<0$, and therefore the infimum on equation \ref{eq:condition_non_increasing_1_5} occurs for the minimum value that $k$ can assume $k=1$, regardless of the value of $N$. Thus, the condition becomes $bc-ad\leq (b-a)d$, which simplifies to 
\begin{equation}
\label{eq:condition_non_increasing_1_6}
c\leq d.
\end{equation}

\item Case $b\geq a$:

This case leads to $(b-a)d\geq0$, and therefore the infimum on equation \ref{eq:condition_non_increasing_1_5} occurs for the largest possible value of $k$, which will be $k=N-1$. Since $(b-a)d$ is non-negative and $k$ is unbounded, the infimum is $0$, and the condition becomes %Since $N$ can be arbitrarily large, the infimum is obtained under $k\to \infty$. The condition becomes 
\begin{equation}
\label{eq:condition_non_increasing_1_7}
bc\leq ad.
\end{equation}

\end{enumerate}

All games contained in these two cases will necessarily have a fixation probability with a decreasing function on $N$.

\end{proof}

\begin{theorem}
If the payoff matrix entries of a $2\times 2$ game satisfy $c\geq d\geq b\geq a$,   or $bc\geq ad$, $c+d\geq 2b$ and $c< d$, 
then the fixation probability $\rho_N$ of a single mutant using $A$ decreases monotonically with $N$.
\end{theorem}

\begin{proof}

As in the proof of Theorem 1, we start by comparing the closed-form expressions for the fixation probabilities (eq. \ref{eq:condition_non_increasing}). We then follow an alternative path, isolating the first term ($j=1$) instead of the last one ($j=N$) in the sum with $N$ elements, and obtain the following relation:
\begin{equation}
\label{eq:condition_non_increasing_2_1}
  \rho_N > \rho_{N+1} \Leftrightarrow \sum_{j=1}^{N-1} \left[ \prod_{k=1}^j \gamma_k^N - \gamma_1^{N+1}\prod_{k=1}^j \gamma_{k+1}^{N+1} \right]<\gamma_1^{N+1}.
\end{equation}

Following the same reasoning as in the previous proof, because $\gamma_1^{N+1}$ is strictly positive, a game meeting the more strict condition
\begin{equation}
\label{eq:condition_non_increasing_2_2}
  \sum_{j=1}^{N-1} \left[ \prod_{k=1}^j \gamma_k^N - \gamma_1^{N+1}\prod_{k=1}^j \gamma_{k+1}^{N+1} \right]\leq 0,
\end{equation}
will necessarily lead to $ \rho_N > \rho_{N+1}$.

By the same reasoning used before, this includes all the payoff matrices that meet $\gamma_k^N\leq\gamma_{k+1}^{N+1}$ and $\gamma_1^{N+1} \geq 1 \forall k\in\{1,...,N-1\}; \forall N\in \{2,3,...\}$. The second part of the condition had to be added to assure that the constant coefficient on the right hand-side product did not make it smaller than the left hand-side product. Applying the form that $\gamma_k^N$ assumes under the frequency-dependent Moran process (equation \ref{eq:relative_fitness_fd_moran}), having $\gamma_k^N\leq\gamma_{k+1}^{N+1}$ and $\gamma_1^{N+1} \geq 1$ is equivalent to
\begin{equation}
\label{eq:condition_non_increasing_2_3}
\begin{dcases}
&a(c-d) \leq (bc-ad)(N-k) \\
&d-c \leq (d-b)N. \\
\end{dcases}
\end{equation}

To fulfill these equations for all  $k\in\{1,...,N-1\}$ and $N\in\{2,3,...\}$, we get the following condition: 
\begin{equation}
\label{eq:condition_non_increasing_2_4}
\begin{dcases}
&bc-ad \geq \sup\left( \left\{ \dfrac{a(c-d)}{N-k}:k\in\{1,...,N-1\};N\in\{2,3,...\} \right\} \right) \\
&d-b\geq \sup\left( \left\{ \dfrac{d-c}{N}:N\in\{2,3,... \}\right\} \right).\\
\end{dcases}
\end{equation}

Parallel to the proof of Theorem 1, the suprema in both equations depend only on the sign of the terms $a(c-d)$ and $d-c$, which are related.

\begin{enumerate}

\item Case $c<d$:

Since $a(c-d)<0$ and $N$ is unbounded, the supremum on the top equation in \ref{eq:condition_non_increasing_2_4} is $0$. Since  $d-c>0$, the one in the bottom equation is the maximum of that expression obtained when $N=2$. Therefore, condition \ref{eq:condition_non_increasing_2_4} is fulfilled if $bc-ad\geq 0$ and $d-b\geq (d-c)/2$, leading to:
\begin{equation}
\label{eq:condition_non_increasing_2_5}
bc\geq ad \hspace{.1cm}\text{,} \hspace{.3cm} c+d\geq 2b.
\end{equation}

\item Case $c\geq d$:

Since $a(c-d)\geq0$, the supremum in the top of equation \ref{eq:condition_non_increasing_2_4} occurs for $k=N-1$. Since $d-c\leq0$ and $N$ is unbounded, the supremum is simply $0$ in the bottom one. Therefore, condition \ref{eq:condition_non_increasing_2_4} turns into $bc-ad\leq a(c-d)$ and $d-b\geq 0$, which lead to:
\begin{equation}
\label{eq:condition_non_increasing_2_6}
b\geq a \hspace{.1cm}\text{,} \hspace{.3cm} d\geq b.
\end{equation}

\end{enumerate}

Any game meeting these conditions will necessarily have a fixation probability with a decreasing function on $N$. 

\end{proof}

Under Theorem 1, there are a total of $7$ orderings of the payoff matrix values represented in table \ref{tab:orderings} that always satisfy the obtained condition. The first condition includes exclusively $6$ complete orderings corresponding to the fixation of strategies from the fourth dominance game listed in table \ref{tab:orderings}, and the two last coordination games. The second condition is always satisfied by ordering $d>b>a>c$ corresponding to the fixation of hare hunters in the Stag Hunt game, and it is partially met for subsets of other $4$ dominance game orderings: $b>d>a>c$, $d>b>c>a$, $d>c>b>a$ and $b>a>d>c$.

Under Theorem 2, the first condition in the union defines exclusively the complete ordering $c>d>b>a$ corresponding to the fixation process of a dominated strategy. The second condition is only satisfied partially under the two orderings $d>b>c>a$ and $d>c>b>a$, both of which are already partially covered by the condition from Theorem 1. Joining Theorems 1 and 2, it can be easily shown that ordering $d>c>b>a$ ends up covered completely by the union of both final conditions, thus always holding decreasing fixation functions.

Additionally, if the dependence of the fixation probability on population size were considered only for large $N$, the second condition in equation \ref{eq:condition_non_increasing_2_4} would simply lead to $d\geq b$. Together with the conditions from Theorem 1, the large $N$ conditions would further cover the complete ordering $d>b>c>a$ as well.

It is worth noting that both Theorems 1 and 2 were proved by analysing the transition probability ratios $\gamma_k^N$ (eq. \ref{eq:transition_prob_ration}), which under the frequency-dependent Moran process become the relative fitness of individuals using B (eq. \ref{eq:relative_fitness_fd_moran}). It can be seen that increasing population size by one individual impacted fixation probabilities in two ways: 1) added an extra term to the sum, and 2) changed ratios $\gamma_k^N$ to $\gamma_k^{N+1}$ in the already existing terms. Contrary to what happens under these orderings, fixation probability functions might not be monotonically decreasing functions of $N$ if the sum of all terms in the denominator of equation \ref{eq:karlin_fixation_probability} decreases with $N$, despite the increase of the number of terms. As will be shown in later sections \ref{sec:fixprob_min}, \ref{sec:fixprob_init_increasing}, and \ref{sec:fixprob_two_extremes}, this is possible under cases where adding individuals to the population decreases the relative fitness $\gamma_k^N$ enough to compensate for the extra individuals that one single mutant would have to replace.

\subsubsection{Games with decreasing functions}

Focusing on the games under which fixation probabilities are decreasing functions of $N$, these might have either a zero or positive asymptotic limit, depending on the values in the payoff matrix. According to \cite{AntalScheuring2006}, dominance games have a mutant strategy holding a positive asymptotic value if $a>c$ and $b>d$, otherwise this value is zero. Looking at the entries in table \ref{tab:orderings} corresponding to these games, we may conclude that the first game ordering for every pair (i.e. the dominating strategy) always holds a positive asymptotic value (DP when decreasing), while the second (i.e. the dominated strategy) always holds a null asymptotic value (D0 when decreasing).

Decreasing fixation probability functions were observed under all dominance game orderings, both for dominating and dominated mutant strategies. Three of the orderings corresponding to the fixation of dominated strategies were proved to only have decreasing functions, %\tb{
and the remaining three were tested for a wide choice of values with no non-decreasing counter-examples found. %}
On the other side, all five dominating mutants which were not proven to hold decreasing fixation functions were found to have examples of fixation probability functions which increased for some value of population size, thus holding alternative shapes.

Regarding anti-coordination games, the borderline case corresponding to switching from having a mutant strategy fixating positively in the asymptotic limit to fixating with null probability was obtained in \cite{AntalScheuring2006} and depends only on the payoff and intensity of selection values. The original expressions were redefined in \cite{SampleAllen2017} based on $I$, the definite integral of $\ln(\gamma(\alpha))$ evaluated between $0$ and $1$, which under $w=1$ becomes
\begin{equation}
\label{eq:mutual_invasion_boundary_condition}
I=\ln \left(\dfrac{c^{\frac{c}{c-d}}  b^{\frac{b}{a-b}}} { d^{\frac{d}{c-d}}  a^{\frac{a}{a-b}} }\right).
\end{equation} 
The function $\gamma(\alpha)$ is the relative fitness $\gamma_k^N$ (eq. \ref{eq:relative_fitness_fd_moran}) under $N\to \infty$, which depends only on the fraction of mutants $\alpha=k/N$ instead of $k$ and $N$ independently.

It was shown that the general borderline case is obtained under $I=0$. If the parameterization of intensity of selection $w$ is considered (see eqs. \ref{eq:fitness_A_w} and \ref{eq:fitness_B_w}), and the weak selection limit taken after the large population limit in exactly this order, condition $I=0$ is equivalent to $a+b=c+d$. Asymptotic fixation probabilities in that case can be represented as first-order terms in $w$ under $a+b\geq c+d$, but only as higher-order terms for $a+b<c+d$.

We observed that under $I>0$, i.e. when the asymptotic fixation probability is null, it is always possible to see decreasing fixation functions (D0) under these games. Otherwise, under $I<0$, i.e. when the mutant strategy fixates positively in the asymptotic limit, the fixation of hawks under the Hawk-Dove game represented the only one of those game orderings which could have strictly decreasing functions (DP) for some choices of the payoff matrix. The remaining $4$ anti-coordination orderings that satisfied $I<0$ for some of their subsets, e.g. fixation of mutant doves under the Hawk-Dove game, were tested systematically and never observed to decrease monotonically to a positive asymptotic value. In section \ref{sec:fixprob_min}, we propose an explanation for why fixation function shapes beyond the decreasing one have not been referred to in most of the previous literature, with the exception of \cite{BroomHadji2010Curve}. %This is so despite the fact they are always observed under some of these game orderings, such as the one which represents the fixation of mutant doves under the Hawk-Dove game.

Under coordination games, asymptotic fixation values are always null independent of the strategy considered as the mutant, as noted by \cite{AntalScheuring2006}. Here we add that five of these orderings (two pairs of complementary orderings and the one representing hare hunters' fixation under the Stag Hunt game), were proven to always have strictly decreasing fixation probabilities (D0). Regarding the other ordering, which corresponds to stag hunters' fixation under the Stag Hunt game, there were striking examples of alternative fixation probability functions which will be mentioned in sections \ref{sec:fixprob_init_increasing} and \ref{sec:fixprob_two_extremes}.

\subsection{\label{sec:fixprob_min}Fixation probability functions with one minimum}

Under some of the explored $2\times 2$ games, there were alternatives to the way the fixation probability of a single mutant depended on population size $N$. One common pattern found in these functions across anti-coordination games and the fixation of dominating strategies started with the occurrence of a plunge for small population sizes until a global minimum of the fixation probability was reached under $N_{min}$. This was followed by a steady increase up to a positive asymptotic value which could be computed following \cite{AntalScheuring2006}.

This type of dependence was not observed under any of the $6$ coordination game orderings nor the other $6$ orderings referring to the fixation of dominated strategies. It was so because all of these fixation probabilities have null asymptotic values. On the other hand, it was sometimes observed under the fixation of dominating strategies, such as $b>d>a>c$, $a>c>b>d$, $a>b>c>d$, $b>a>d>c$, and $b>a>c>d$. This included the fixation of defectors under the Prisoner's Dilemma (e.g. under payoffs $[2,4,1.9,2.1]$). The occurrence of this dependence under this set of games is absent from the previous literature.

Nonetheless, it was  under anti-coordination games that we found this behaviour to be pervasive. %This was present under all but one anti-coordination game orderings. The game ordering found to be an exception ($c>b>d>a$) never meets $I<0$, a criterion necessary to have a positive asymptotic value and thus observe this profile of fixation probability. On top of this, a
This was present under all $5$ anti-coordination game orderings which allowed positive asymptotic values. As already noted in section \ref{sec:fixprob_decreasing}, $4$ out of these $5$ anti-coordination game orderings didn't accommodate decreasing functions of population size, suggesting that the existence of increasing regions is the norm rather than the exception under these games.

Despite being pervasive under these games, which have been thoroughly studied in the past \cite{MaynardSmith1973,HauertDoebeli2004,DoebeliHauert2005,BroomRychtar2013,Broom2012Framework,Broom2015AInteractions}, this effect did not receive  much attention, with the exception of \cite{BroomHadji2010Curve}. A significant amount of approaches to games in finite populations have considered the weak selection limit \cite{Traulsen2006,Wild2007WeakSelection,SampleAllen2017} and show that when selection tends to zero fast enough (when compared to the increase of population size), we should expect fixation probability functions to always decrease asymptotically. Motivated by this, we have looked at the impact of decreasing intensity of selection on the fixation probability functions, represented in figure \ref{fig:fixprob_DUP_nw}. For a fixed choice of payoffs $[a,b,c,d]$, decreasing the intensity of selection $w$ pushes the population size from which fixation starts to increase $N_{min}$ to larger values. In fact, they are approximately inversely proportional, $N_{min}(w)\sim 1/w$ (see appendix \ref{appendix:weak_selection}).

\begin{figure}[ht]
\centering
\includegraphics[width=0.65\textwidth]{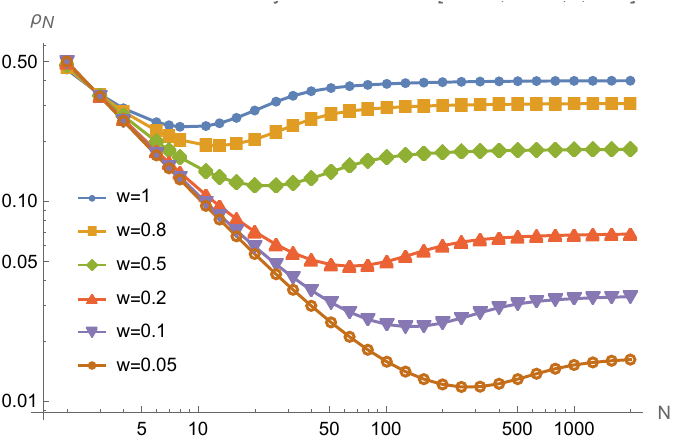}
\caption[Single-mutant fixation probability functions and the effect of weak selection (low intensity of selection $w$) on pushing minima under shape DUP to higher population sizes $N$.]{\label{fig:fixprob_DUP_nw}
Single-mutant fixation probability functions and the effect of weak selection on moving global minima under shape DUP to higher population sizes. The fixation probabilities shown were obtained for the fixation of doves under the Hawk-Dove game (HD) with $[5.5,5,6,3]$.}
\end{figure}

These results suggest that even if we consider an arbitrarily small value of the intensity of selection (i.e. weak selection limit), we may still see fixation probabilities increase with population size for large enough populations $N>N_{min}(w)$ (i.e. in the large population limit). This should not contradict the results obtained in \cite{SampleAllen2017}, since in the scenario we are describing the large population limit should dominate over the weak selection one. Additionally, because turning points $N_{min}(w)$ become very large under weak selection, fixation probabilities could become less representative of what happens under anti-coordination games in finite populations. Even though pure states are absorbing ones, when mutations are considered the population might spend most of the evolutionary time in particular transient states \cite{AntalScheuring2006,Vasconcelos2017Hierarchic}, which are called quasi-stationary states in those cases \cite{Zhou2010,Overton2022QuasiStationarySIS,Nasell1999QuasiStationary1,Nasell1999QuasiStationary2}. In appendix \ref{appendix:weak_selection}, we show that average conditional fixation times increase polynomially with $N_{min}(w)$ as they do with $N$ under neutral fixation, instead of increasing  exponentially as is predicted for anti-coordination games under a fixed intensity of selection \cite{AntalScheuring2006}. These results highlight why it is essential to beware of the impact of the order in which the weak selection limit $w\to 0$ and the large population limit $N\to \infty$ are considered under evolutionary dynamics, as it has been clearly stated in \cite{SampleAllen2017}.

\subsection{\label{sec:fixprob_init_increasing}Fixation probability functions increasing under small populations}

Under a restricted number of games of all three types (dominance, anti-coordination, coordination), fixation probabilities were observed to increase with population size $N$ for small populations. The particular case where $\rho_3>\rho_2$ holds is an extreme one where this may happen. However, the condition represented by that is simple enough to allow an analytical approach. 

\begin{theorem}
If the payoff matrix entries of a $2\times 2$ game satisfy $c>d$ and $a+b>2c(c+d)/(c-d)$, then $\rho_3>\rho_2$.
\end{theorem}

\begin{proof}

The fixation probability values under these population sizes are the following:
\begin{equation}
\label{eq:condition_increasing2_1}
  \rho_2=\frac{1}{1+ \gamma_1^2},
\end{equation}
\begin{equation}
\label{eq:condition_increasing2_2}
  \rho_3=\frac{1}{1+ \gamma_1^3 + \gamma_1^3 \cdot \gamma_2^3}.
\end{equation}

These equations lead to the equivalence
\begin{equation}
\label{eq:condition_increasing2_3}
\rho_3 > \rho_2 \Leftrightarrow
\gamma_1^3 + \gamma_1^3 \cdot \gamma_2^3< \gamma_1^2 \Leftrightarrow \gamma_1^3 \cdot  \left(1 + \gamma_2^3\right)< \gamma_1^2.
\end{equation}

Applying the definition of $\gamma_k^N$ under the frequency-dependent Moran process (equation \ref{eq:relative_fitness_fd_moran}) on equation \ref{eq:condition_increasing2_3}, we get that the following is a sufficient condition:
\begin{equation}
\label{eq:condition_increasing2_4}
  a+b > \lambda \hspace{0.1cm} \text{and} \hspace{0.1cm} c>d,
\end{equation}
where we have used $\lambda$, defined as
\begin{equation}
\label{eq:condition_increasing2_definition}
 \lambda = 2 \hspace{2mm}\frac{c\hspace{1mm}(c+d)}{c-d}.
\end{equation}

\end{proof}

This condition is only possible under $6$ out of the $24$ possible payoff orderings. There are twelve orderings where $c>d$ but only in six of them it is possible to have $a+b>\lambda$. These are the ones where $c>d$ and either $a$ or $b$ is the largest entry in the payoff matrix. These orderings are itemised in the last column of table \ref{tab:orderings}.

Condition \ref{eq:condition_increasing2_3} highlights the conflict of having one extra individual in a population. One more resident leads to mutants having another individual to replace. Thus the presence of an extra term on the left hand-side of equation \ref{eq:condition_increasing2_3}. In order to observe fixation probabilities increasing from population size $N=2$ to $N=3$, the relative fitness of residents needs to be much lower for higher proportions of them in the population.

There was a total of three shapes of fixation probability functions found that met $\rho_3>\rho_2$. Fixation probabilities might increase monotonically up to a positive asymptotic value (UP), they might increase up to a maximum value and then decrease down to a positive value (UDP), or they might increase up to a maximum and then decrease to zero (UD0). These three shapes are represented in the summary presented in figure \ref{fig:fix_prob}.

Three of the orderings that satisfy this condition correspond to the fixation of dominating strategies. This means that the fixation probability's asymptotic value under all of these cases is positive. Under the two orderings $a>c>b>d$ and $a>b>c>d$ both UP and UDP were found depending on the particular parameter choices, while under $b>a>c>d$ only UP was found.

Under the two anti-coordination game orderings which are able to satisfy these conditions, the only shape found from these three was UP. %\tb{
Systematic checks suggest that equation  $I>0$ is either not satisfied at all, or at least at the same time as equation \ref{eq:condition_increasing2_4} under those two orderings. %} 
This means that when fixation probabilities increase from $N=2$ to $N=3$ under those two orderings, we should only observe functions with positive asymptotic fixation probabilities (i.e. UD0 should never be observed there).

Fixation probability functions with that shape -- having one maximum and then tending to zero (UD0) -- were found only for the fixation of stag hunters under the Stag Hunt game. This is the case exhibited in figure \ref{fig:fix_prob_null}. %\tb{
As suggested by the tests mentioned in the previous paragraphs, there was no other ordering where UD0 was observed, therefore establishing the particularity of this game and justifying further interest in its study.%}

\subsection{\label{sec:fixprob_two_extremes}Fixation probability functions with two extremes}

Under orderings corresponding to games of all three types, fixation probability functions $\rho_N$ were sometimes observed to decrease with population size $N$ for small populations, have an increasing region for intermediate population sizes and finally decrease again under large populations. These increasing regions were necessarily delimited by two local extremes, one minimum and one maximum.

%corresponding to anti-coordination games or the fixation of the stag hunters 

%or the fixation of the stag hunters

%If this effect was strong enough, it was seen to come with the existence of two local extremes - one minimum followed by a maximum (DUD0).

%This can be understood by analysing the payoff matrices under which it is observed. 

Under anti-coordination games, these regions were observed only when asymptotic values were null (DUD0) and occurred when the choice payoffs led to a positive value of $I$ very close to zero (see equation \ref{eq:mutual_invasion_boundary_condition}). %Close to this boundary, %if $I<0$ then a mutant using $A$ fixates with finite probability under the limit $N\to \infty$. If $I>0$ fixation becomes zero in the limit $N\to \infty$. 
If $I>0$ then a mutant using $A$ has a null probability of  fixation in the limit $N\to \infty$. However, close to the border $I=0$ and for small enough populations, they may fixate similarly as if they were on the negative side. In the particular case in figure \ref{fig:DUD0}, $N = 200$ is enough for $I\leq 10^{-3}$. The occurrence of two extremes and the increasing regions to which they lead present as transitional features between the decreasing shape D0 (seen under a high enough positive $I$), and the asymptotically positive DUP (seen under a negative $I$). Additionally, as we get further closer to $I\to0^+$, we see the positive values of $\rho_N$ breaking down to zero only for larger and larger population sizes.

\begin{figure}[H]
\centering
\includegraphics[width=0.55\textwidth]{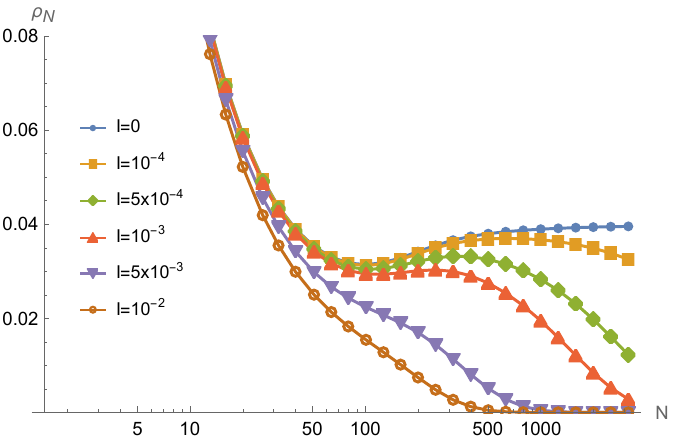}
\caption[Single-mutant fixation probability functions transitioning from having shape DUP to D0 through DUD0 under Hawk-Dove game.]{\label{fig:DUD0}
Single-mutant fixation probability functions transitioning from having shape DUP ($I=0$) to D0 ($I=\{5\times 10^{-3}, 10^{-2}\}$) through DUD0 ($I=\{10^{-4}$, $5\times 10^{-4}, 10^{-3}\}$). Fixation probabilities were obtained for games with payoff parameters $[5.5,5,6,d]$ and population size $N$. The values of $d$ were calculated from each chosen value of $I$ (eq. \ref{eq:mutual_invasion_boundary_condition}). Their choice maintains the game under ordering $c>a>b>d$, i.e. fixation of doves under Hawk-Dove game (HD), and reproduces the approach of limit $I\to0^+$. The approximate values of $d$ are, in order, $d=4.530$ ($I=0$), $d=4.531$, $d=4.535$, $d=4.540$, $d=4.581$, and $d=4.631$.}
\end{figure}

This shape was also observed in the context of the fixation of stag hunters in the Stag Hunt game when function shapes transitioned between D0 and UD0, i.e. between being always decreasing and having a global maximum (instead of a global minimum as in the previous case). This is shown in figure \ref{fig:DUD0_stag}.

\begin{figure}[h!]
\centering
\includegraphics[width=0.55\textwidth]{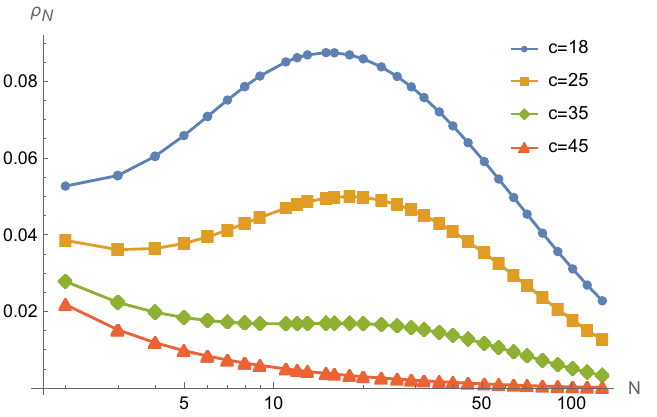}
\caption[Single-mutant fixation probability functions transitioning from having shape UD0 to D0 through DUD0.]{\label{fig:DUD0_stag}Single-mutant fixation probability functions transitioning from having shape UD0 ($c=18$) to D0 ($c=45$) through DUD0 ($c=\{25,35\}$). Fixation probabilities were obtained for games with payoff parameters $[50,1,c,2]$ and population size $N$. The values of $c$ were chosen as to always being under ordering $a>c>d>b$ and allowing condition $\rho_3>\rho_2$ (eq. \ref{eq:condition_increasing2_4}) to be met only for the lowest value $c=18$.}
\end{figure}

A parallel fixation probability function having an intermediate population size increasing region between two extremes, but with a positive asymptotic value (DUDP) was observed under the fixation of two dominating strategies. This can be seen both in the summarising figure \ref{fig:fix_prob_finite_2} and in figure \ref{fig:DUDP}. This shape was observed under orderings $a>c>b>d$ and $a>b>c>d$, as a transition between functions with one global minimum (DUP) and functions which start increasing up to a global maximum and then decrease to a positive value (UDP). Those are the only two orderings that may satisfy $\rho_{3}>\rho_{2}$ by taking the form of UDP.

\begin{figure}[ht]
\centering
\includegraphics[width=0.6\textwidth]{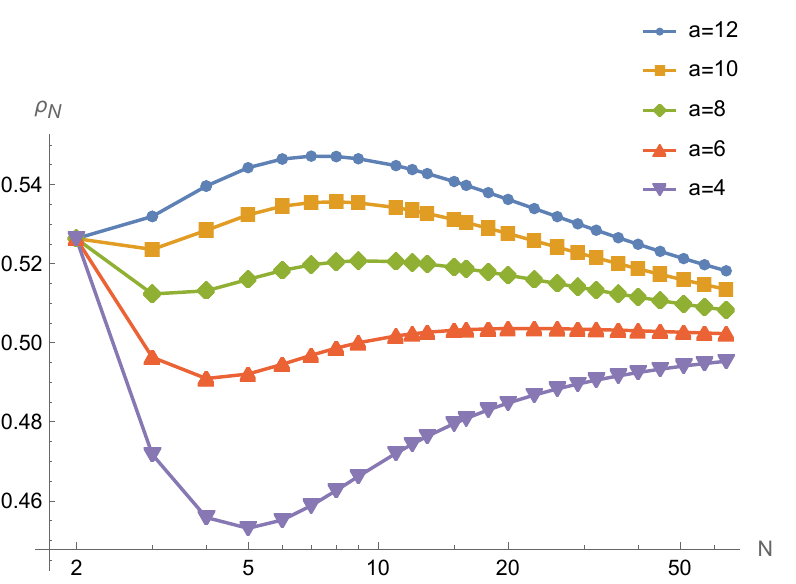}
\caption[Single-mutant fixation probability functions transitioning from having shape UDP to DUP through DUDP.]{\label{fig:DUDP}Single-mutant fixation probability functions transitioning from having shape UDP ($a=12$) to DUP ($a=4$) through DUDP ($a=\{10,8,6\}$). Fixation probabilities were obtained for games with payoff parameters $[a,2,1.8,1]$ and population size $N$. The values of $a$ were chosen as to always being under ordering $a>b>c>d$ and allowing condition $\rho_3>\rho_2$ (eq. \ref{eq:condition_increasing2_4}) to be met only for the highest value $a=12$.}
\end{figure}

These results suggest that the shapes DUDP and DUD0 having two extremes delimiting an increasing intermediate population region occur for transitions between fixation functions with one global extreme, either a maximum or a minimum, and other function shapes. Functions with one global extreme are characterized by having different trends (increasing/decreasing) for small and large population sizes. If incrementally changing the payoff matrix values switches the increasing trend to a decreasing one, an increasing region may emerge under intermediate populations.

\section{Discussion}

Understanding fixation processes and the probability of their occurrence is essential to the analysis of evolutionary games on finite populations. In this context, \cite{Taylor2004} observed that even if individuals are interacting via the same game, the size $N$ of the population where they are included could lead to different evolutionarily stable strategies $ESS_N$ as defined in \cite{Nowak2004}. %\tb{
This is an emerging feature of evolutionary processes under frequency-dependent fitness since stability under the fixed fitness Moran process never depends on population size. %} 
Here, we provided an extensive analysis of the dependence of single-mutant fixation probabilities on population size, focusing on the simplest $2\times 2$ games in well-mixed populations.

We have proved that under nine game orderings, fixation probabilities should be strictly decreasing functions. %\tb{
Under three other game orderings, extensive tests were performed without finding any other form of dependence beyond this one. %} 
Altogether, these games included the fixation of all six dominated strategies, one dominating strategy, and almost all strategies under coordination game orderings (five out of six).

However, a lot of interesting dependencies on population size $N$ arise in the remaining twelve game orderings, under which single-mutant fixation functions were observed to have intervals of population sizes for which they increased with $N$. As counter-intuitive as this may seem, populations having extra individuals whom a single mutant would have to replace during fixation might actually increase that mutant's chances of fixating. 

Under large populations, this happened when the mutant's relative fitness was higher under lower proportions of its type in the population, i.e. when mutants did relatively better among residents than they did among their own type: $d/b<c/a$. In this limit, the selection dynamics are completely characterized by these two ratios, as it was already noted in \cite{Taylor2004}. Increasing population size guaranteed that both mutants and residents interacted more frequently with residents in the initial steps of the fixation process. Therefore, the increase in the number of resident individuals to be replaced by one mutant was offset by the higher replacement probabilities.

Surprisingly, fixation functions were observed to increase under cases where mutants did relatively worse among residents than among themselves for large populations. This was observed under $3$ game orderings, where fixation probabilities increased for small populations but eventually decreased to infinity. It was also noted in \cite{Taylor2004} that payoffs $b$ and $c$ are the only relevant entries of the payoff matrix under $N=2$. Thus, under small populations, increasing population size means increasing the maximum frequency of individuals of the same type that are found in the population before fixation occurs (e.g. $0\%$ for $N=2$, $50\%$ for $N=3$, $80\%$ for $N=6$). Increasing population size thus increases the frequency of interactions between equal types. Having high enough $a>\lambda-b$ and low enough $d<c$ allowed for fixation probabilities to increase for the smallest population sizes, even when $d/b>c/a$.

The confluence of all these effects in the $12$ remaining orderings resulted in a high diversity of shapes of single-mutant fixation probability functions, from functions with one minimum (DUP), to always increasing functions (UP), initially increasing but asymptotically decreasing ones (UDP and UD0), and functions with two extremes (DUDP and DUD0). We tried to understand under which orderings, and why, these function shapes emerged.

Anti-coordination games, such as Hawk-Dove \cite{MaynardSmith1973,BroomRychtar2013} (also called Snowdrift \cite{HauertDoebeli2004,DoebeliHauert2005}) were shown to hold a wide variety of fixation probability function shapes. One of the most striking observations was the pervasiveness of functions with one minimum, a shape already noted in \cite{BroomHadji2010Curve} under the fixation of dove strategy. In four orderings, this was the only shape observed when asymptotic values were positive, even for arbitrarily low values of intensity of selection. While it was shown in \cite{SampleAllen2017} that fixation probabilities decrease with population size when the weak selection limit is dominant over the large population limit, our results suggest that if the considered limit order is the reverse, weak selection does not necessarily erase these non-decreasing shapes. Additionally, when approaching the limit where fixation functions change from being asymptotically null to having this shape with a global minimum, all anti-coordination game orderings showed functions with two extremes. Finally, two of these orderings also showed monotonically increasing shapes when $b$ was sufficiently large, i.e. when mutants benefited the most from interacting with residents. 

The six game orderings associated with the fixation of dominating strategies \cite{Taylor2004} (also named unbeatable \cite{Hamilton1967,Nowak2004}), such as defectors under the Prisoner's Dilemma \cite{Axelrod1984,Poundstone1992}, were observed to hold monotonically decreasing functions, with one complete ordering and subsets of some of the remaining being comprehended in the conditions of Theorems 1 and 2. However, they were also shown to have a wide variety of other function shapes. Functions with one global minimum were observed under five out of the six orderings. Under three of these orderings, we observed monotonically increasing functions. Two of them additionally showed alternative functions keeping an initial increase but decreasing for larger population sizes, and also functions with two extremes when transitioning between those with a global minimum and the previous shape. On the other side, the fixation of dominated strategies, such as cooperation under the same games, was proved or tested to systematically hold monotonically decreasing functions.

In the context of coordination games, we have proved that under five out of the six possible orderings, fixation functions were always monotonically decreasing. However, the fixation of stag hunters (the reward-dominant strategy under the Stag Hunt game \cite{Skyrms2001,Skyrms2004}) was observed to have exceptional results in this context. If the reward for cooperation $a$ is large enough, we observe initially increasing functions that then tend to zero thus holding a maximum. This shows that there is an optimal population size for the fixation of a single stag hunter in those cases. Functions with two extremes were observed under these games, when transitioning between the monotonically decreasing functions and the ones with a global maximum, leading to fixation probabilities holding similar values for a wide range of population sizes.

There are numerous relevant borderline cases where finite populations might hold relevant differences. One example leading to particularly different results in finite populations is the transition through $I=0$ (see \cite{SampleAllen2017}), which in section \ref{sec:fixprob_two_extremes} was shown to lead to functions with two extremes. However, the orderings were explored assuming the entries in the payoff matrix to be different pairwisely. This ignores potentially interesting equal-payoff cases, such as the one with $a=c$ and $b=d$ under which there is neutral selection for infinitely large populations but not for finite ones, leading to size-dependent equilibria \cite{Taylor2004}.

We have verified that considering %the death-birth \cite{Ohtsuki2006ANetworks} and 
the evolutionary dynamics rising from the pairwise comparison \cite{Traulsen2006} can lead to all eight fixation function shapes reported here. Additionally, considering multiplayer social dilemmas in well-mixed populations would allow transition probability ratios $\gamma_k^N$ to have a non-monotonic dependence on $k$ \cite{Wu2013MultiplayerGames,Broom2019SocialDilemmas}, which could lead to the emergence of new fixation probability functions of population size. Finally, fixation probabilities have been studied in the context of structured populations as well, often holding distinctive dependencies on population size \cite{BroomHadji2010Curve,Hadji2011StarDynamics,Hindersin2016FixationGraphs,Allen2021FixationGraphWeakSelection}. Expanding our analysis to structured populations could hold new effects and lead to a more general theory of fixation probability functions.

\appendix

\section{\label{appendix:weak_selection}The Effect of the Weak Selection Limit}

In section \ref{sec:fixprob_min}, it was observed that anti-coordination games held fixation probabilities which increased for population sizes above a turning point $N_{min}(w)$ (see figure \ref{fig:fixprob_DUP_nw}). Figure \ref{fig:nw_nmin_w} suggests that turning points are inversely proportional to the intensity of selection $N_{min}(w)\sim 1/w$, which means that they would become increasingly large under the weak selection limit. Under this limit, it could be argued that the increase in the population sizes for which fixation probabilities start to increase would lead to a loss of significance of fixation processes due to most of the evolutionary time being spent in transient/mixed states \cite{AntalScheuring2006,Vasconcelos2017Hierarchic}, also called quasi-stationary states \cite{Zhou2010,Overton2022QuasiStationarySIS,Nasell1999QuasiStationary1,Nasell1999QuasiStationary2}.

    \begin{figure}[ht]
    \centering
    
    \begin{subfigure}[b]{0.48\textwidth}
    \centering
    \includegraphics[width=\textwidth]{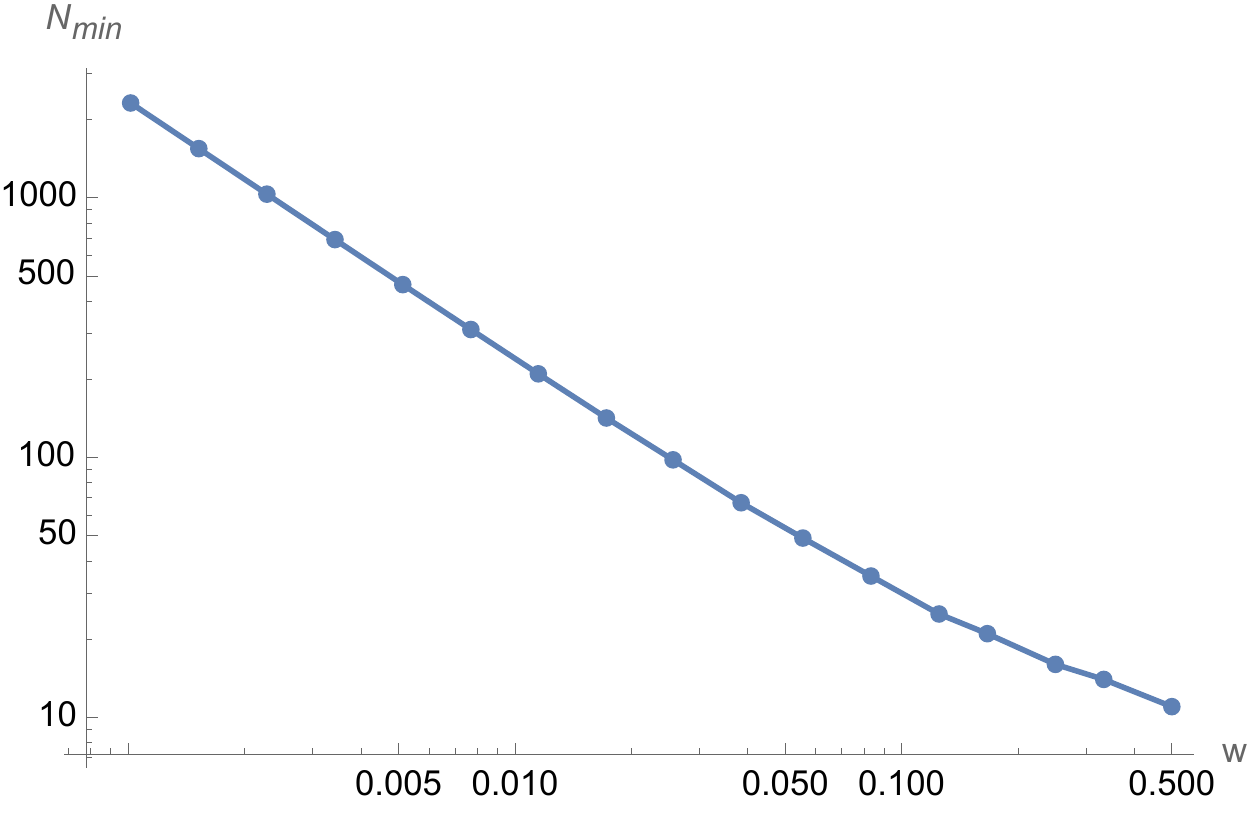}
    \caption{\label{fig:nw_nmin_w}}
    \end{subfigure}
    \begin{subfigure}[b]{0.51\textwidth}
    \centering
    \includegraphics[width=\textwidth]{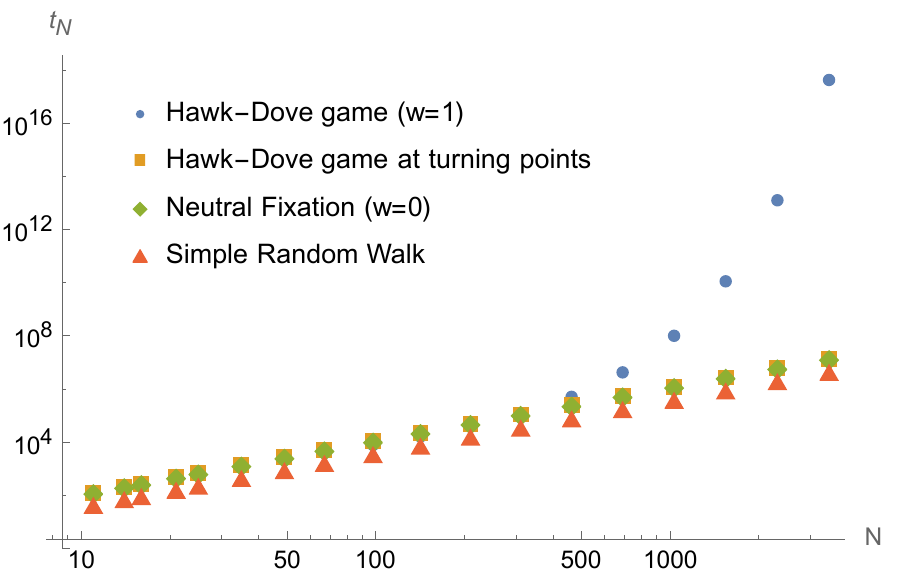}
    \caption{\label{fig:nw_times}}
    \end{subfigure}
    
    \caption[Varying]{\label{fig:nw}Fixation process under game with payoff values $[5.5,5,6,3]$, corresponding to dove fixation in the Hawk-Dove game (HD). Figure \ref{fig:nw_nmin_w} exhibits the turning point $N_{min}$ as a function of the intensity of selection $w$, showing that they have an approximately inversely proportional relation. In figure \ref{fig:nw_times}, the average conditional fixation times $t_N$ were shown for (1) the Hawk-Dove game with the maximum intensity of selection ($w=1$ and independent $N$), (2) the Hawk-Dove game at the turning points of figure \ref{fig:nw_nmin_w} ($w\in(0,1)$ and $N=N_{min}(w)$), (3) the neutral fixation scenario ($w=0$ and independent $N$), and (4) the simple (symmetric) random walk.}
    \end{figure}

The average conditional fixation time is the average number of discrete steps that the population takes to go from state $i=1$ to state $i=N$ (i.e. the fixation of a single mutant) conditional on it happening:
\begin{equation}
    \label{eq:average_conditional_fixation_time_general}
    t_N=\dfrac{\sum_{t=0}^{\infty}t\phi_N(t)}{\sum_{t=0}^{\infty}\phi_N(t)},
\end{equation}
where $\phi_N(t)$ is the probability that the population gets from state $i=1$ to state $i=N$ after $t$ discrete steps. It should be clear from this definition that $\sum_{t=0}^{\infty}\phi_N(t)=\rho_N$. Following \cite{AntalScheuring2006}, we can obtain a recursive relation leading to the following expression (see \cite{DellaRossa2017Unifying,Huang2018} for alternative expressions):
\begin{equation}
    \label{eq:average_conditional_fixation_time}
    t_N=\sum_{n=1}^{N-1}\dfrac{s_{0,n-1}s_{n,N-1}}{P_{n^+}^Nq_ns_{0,N-1}}, \hspace{12pt} s_{n,m}=\sum_{k=n}^{m}q_k, \hspace{12pt} q_n=\prod_{j=1}^{n}\gamma_j^N.
\end{equation}

We are interested in comparing the average time obtained in the turning points $N_{min}(w)$ of figures \ref{fig:fixprob_DUP_nw} and \ref{fig:nw_nmin_w} with three different scenarios: the simple (symmetric) random walk, the neutral fixation case and the same game for maximum intensity of selection $w=1$.

The simple (symmetric) random walk is obtained by considering the stochastic process under which the transition between a state $i$ and one of its two neighbours is one-half: $P_{i^+}=P_{i^-}=1/2$. The average conditional fixation time under this process is
\begin{equation}
\label{eq:average_conditional_fixation_time_random_walk}
    t_N=\dfrac{1}{3}(N^2-1).
\end{equation}

On the other hand, under the frequency-dependent Moran process the transition probabilities $P_{i^+}$ and $P_{i^-}$ do not add up to $1$. This is so because at each time step there is a probability that the population will remain in the same state $P_{i^=}$. Thus, when we are under neutral selection ($w=0$), equation \ref{eq:average_conditional_fixation_time} leads to a different equation than \ref{eq:average_conditional_fixation_time_random_walk}, which is three times slower than the simple random walk
\begin{equation}
\label{eq:average_conditional_fixation_time_neutral_fixation}
    t_N=N(N-1).
\end{equation}

By analysing fixation processes under $2\times2$ games, \cite{AntalScheuring2006} concluded that anti-coordination games have average conditional fixation times that grow exponentially with population size for asymptotically large populations. The base of exponential growth $\tau$ can be calculated from the payoff matrix. The average conditional fixation time thus becomes
\begin{equation}
\label{eq:average_conditional_fixation_time_mutual_invasion}
    t_N\sim \tau(a,b,c,d)^N,
\end{equation}
where $\tau(a,b,c,d)>1$ for any choice of $a$, $b$, $c$, and $d$ leading to an anti-coordination game with $I<0$.

In figure \ref{fig:nw_times}, we observe that as intensity of selection decreases, and turning points increase accordingly $N_{min}(w)\to\infty$, fixation times under this anti-coordination game keep on being of the order of the neutral ones $t_N \sim N^2$ instead of increasing exponentially with $N$. As noted in \cite{SampleAllen2017}, considering limits $N\to \infty$ and $w\to 0$ in different orders may lead to different fixation outcomes. Here we present a case where it is needed to clarify that relation in order to understand what happens asymptotically.

\section*{Acknowledgements}

This project has received funding from the European Union's Horizon 2020 research and innovation programme under the
Marie Skłodowska-Curie grant agreement No 955708.

\bibliography{references}

\begin{thebibliography}{10}

\bibitem{MaynardSmith1991Signalling}
John Maynard~Smith.
\newblock {Honest signalling: the Philip Sidney game}.
\newblock {\em Animal Behaviour}, 42(6):1034--1035, 12 1991.

\bibitem{Skyrms2010Signalling}
Brian Skyrms.
\newblock {\em {Signals: Evolution, Learning, {\&} Information}}.
\newblock Oxford University Press, 2010.

\bibitem{Hutteger2014Signalling}
Simon Huttegger, Brian Skyrms, Pierre Tarres, and Elliott~O. Wagner.
\newblock {Some dynamics of signaling games}.
\newblock {\em Proceedings of the National Academy of Sciences of the United
  States of America}, 111(SUPPL.3):10873--10880, 7 2014.

\bibitem{Axelrod1984}
Robert Axelrod.
\newblock {\em {The Evolution of Cooperation}}.
\newblock Basic Books, 1984.

\bibitem{Poundstone1992}
William Poundstone.
\newblock {\em {Prisoner's Dilemma: John von Neumann, Game Theory and the
  Puzzle of the Bomb}}.
\newblock Anchor Books, New York, USA, 1992.

\bibitem{Skyrms2004}
Brian Skyrms.
\newblock {\em {The stag hunt and the evolution of social structure}}.
\newblock Cambridge University Press, 2004.

\bibitem{MaynardSmith1973}
John Maynard~Smith and George~R Price.
\newblock {The Logic of Animal Conflict}.
\newblock {\em Nature}, 246:15--18, 1973.

\bibitem{BroomRychtar2013}
{Mark Broom} and {Jan Rycht{\'{a}}{\v{r}}}.
\newblock {\em {Game-Theoretical Models in Biology}}.
\newblock Chapman {\&} Hall/ CRC Press, London, UK, 1st edition, 2013.

\bibitem{Taylor1978}
Peter~D Taylor and Leo~B Jonker.
\newblock {Evolutionarily Stable Strategies and Game Dynamics}.
\newblock {\em Mathematical Biosciences}, 40:145--156, 1978.

\bibitem{Hofbauer1998}
Josef Hofbauer and Karl Sigmund.
\newblock {\em {Evolutionary Games and Population Dynamics}}.
\newblock Cambridge University Press, 1998.

\bibitem{Nowak2004}
Martin~A. Nowak, Akira Sasaki, Christine Taylor, and Drew Fudenberg.
\newblock {Emergence of cooperation and evolutionary stability in finite
  populations}.
\newblock {\em Nature}, 428:646--650, 2004.

\bibitem{Moran1958}
Patrick Alfred~Pierce Moran.
\newblock {Random processes in genetics}.
\newblock {\em Mathematical Proceedings of the Cambridge Philosophical
  Society}, 54(1):60--71, 1958.

\bibitem{Taylor2004}
Christine Taylor, Drew Fudenberg, Akira Sasaki, and Martin~A. Nowak.
\newblock {Evolutionary game dynamics in finite populations}.
\newblock {\em Bulletin of Mathematical Biology}, 66(6):1621--1644, 11 2004.

\bibitem{DellaRossa2017Unifying}
Fabio Della~Rossa, Fabio Dercole, and Cristina Vicini.
\newblock {Extreme Selection Unifies Evolutionary Game Dynamics in Finite and
  Infinite Populations}.
\newblock {\em Bulletin of Mathematical Biology}, 79(5):1070--1099, 5 2017.

\bibitem{Huang2018}
Feng Huang, Xiaojie Chen, and Long Wang.
\newblock {The Equivalence Induced by Unifying Fitness Mappings in
  Frequency-Dependent Moran Process}.
\newblock In {\em Proceedings of the 2018 IEEE 7th Data Driven Control and
  Learning Systems Conference (DDCLS'18)}, pages 847--853, Enshi, China, 2018.

\bibitem{Schaffer1988ESSN}
Mark~E. Schaffer.
\newblock {Evolutionarily stable strategies for a finite population and a
  variable contest size}.
\newblock {\em Journal of Theoretical Biology}, 132(4):469--478, 6 1988.

\bibitem{deSouza2019plotshape}
Evandro~P. de~Souza, Eliza~M. Ferreira, and Armando G.~M. Neves.
\newblock {Fixation probabilities for the Moran process in evolutionary games
  with two strategies: graph shapes and large population asymptotics}.
\newblock {\em Journal of Mathematical Biology}, 78(4):1033--1065, 3 2019.

\bibitem{Traulsen2006}
Arne Traulsen, Martin~A. Nowak, and Jorge~M. Pacheco.
\newblock {Stochastic dynamics of invasion and fixation}.
\newblock {\em Physical Review E - Statistical, Nonlinear, and Soft Matter
  Physics}, 74(1), 2006.

\bibitem{Traulsen2007}
Arne Traulsen, Jorge~M. Pacheco, and Martin~A. Nowak.
\newblock {Pairwise comparison and selection temperature in evolutionary game
  dynamics}.
\newblock {\em Journal of Theoretical Biology}, 246(3):522--529, 6 2007.

\bibitem{BroomHadji2010Curve}
Mark Broom, Christophoros Hadjichrysanthou, and Jan Rycht{\'{a}}{\v{r}}.
\newblock {Evolutionary games on graphs and the speed of the evolutionary
  process}.
\newblock {\em Proceedings of the Royal Society A: Mathematical, Physical and
  Engineering Sciences}, 466(2117):1327--1346, 5 2010.

\bibitem{Hadji2011StarDynamics}
Christophoros Hadjichrysanthou, Mark Broom, and Jan Rycht{\'{a}}{\v{r}}.
\newblock {Evolutionary Games on Star Graphs Under Various Updating Rules}.
\newblock {\em Dynamic Games and Applications}, 1(3):386--407, 9 2011.

\bibitem{Hindersin2016FixationGraphs}
Laura Hindersin, Marius M{\"{o}}ller, Arne Traulsen, and Benedikt Bauer.
\newblock {Exact numerical calculation of fixation probability and time on
  graphs}.
\newblock {\em BioSystems}, 150:87--91, 12 2016.

\bibitem{Allen2021FixationGraphWeakSelection}
Benjamin Allen, Christine Sample, Patricia Steinhagen, Julia Shapiro, Matthew
  King, Timothy Hedspeth, and Megan Goncalves.
\newblock {Fixation probabilities in graph-structured populations under weak
  selection}.
\newblock {\em PLoS Computational Biology}, 17(2), 2 2021.

\bibitem{Fudenberg2006SML}
Drew Fudenberg, Martin~A. Nowak, Christine Taylor, and Lorens~A. Imhof.
\newblock {Evolutionary game dynamics in finite populations with strong
  selection and weak mutation}.
\newblock {\em Theoretical Population Biology}, 70(3):352--363, 11 2006.

\bibitem{Hauert2007SMLSimulations}
Christoph Hauert, Arne Traulsen, Hannelore Brandt, Martin~A. Nowak, and Karl
  Sigmund.
\newblock {Via Freedom to Coercion: The Emergence of Costly Punishment}.
\newblock {\em Science}, 316(5833):1905--1907, 6 2007.

\bibitem{SegbroeckSantos2009Diversity}
Sven Van~Segbroeck, Francisco~C. Santos, Tom Lenaerts, and Jorge~M. Pacheco.
\newblock {Reacting Differently to Adverse Ties Promotes Cooperation in Social
  Networks}.
\newblock {\em Physical Review Letters}, 102(5), 2 2009.

\bibitem{Santos2011PrePlaySignalling}
Francisco~C. Santos, Jorge~M. Pacheco, and Brian Skyrms.
\newblock {Co-evolution of pre-play signaling and cooperation}.
\newblock {\em Journal of Theoretical Biology}, 274(1):30--35, 4 2011.

\bibitem{Wagner2020Signalling}
Elliott~O. Wagner.
\newblock {Conventional Semantic Meaning in Signalling Games with Conflicting
  Interests}.
\newblock {\em The British Journal for the Philosophy of Science},
  66(4):751--773, 12 2020.

\bibitem{Wu2012SML}
Bin Wu, Chaitanya~S. Gokhale, Long Wang, and Arne Traulsen.
\newblock {How small are small mutation rates?}
\newblock {\em Journal of Mathematical Biology}, 64(5):803--827, 4 2012.

\bibitem{Nasell1999QuasiStationary1}
Ingemar Nasell.
\newblock {On the time to extinction in recurrent epidemics}.
\newblock {\em Journal of the Royal Statistical Society: Series B (Statistical
  Methodology)}, 61(2):309--330, 1999.

\bibitem{Nasell1999QuasiStationary2}
Ingemar Nasell.
\newblock {On the quasi-stationary distribution of the stochastic logistic
  epidemic}.
\newblock {\em Mathematical Biosciences}, 156:21--40, 1999.

\bibitem{Zhou2010}
Da~Zhou, Bin Wu, and Hao Ge.
\newblock {Evolutionary stability and quasi-stationary strategy in stochastic
  evolutionary game dynamics}.
\newblock {\em Journal of Theoretical Biology}, 264(3):874--881, 6 2010.

\bibitem{Vasconcelos2017Hierarchic}
Vítor~V. Vasconcelos, Fernando~P. Santos, Francisco~C. Santos, and Jorge~M.
  Pacheco.
\newblock {Stochastic Dynamics through Hierarchically Embedded Markov Chains}.
\newblock {\em Physical Review Letters}, 118(5), 2 2017.

\bibitem{AntalScheuring2006}
Tibor Antal and István Scheuring.
\newblock {Fixation of strategies for an evolutionary game in finite
  populations}.
\newblock {\em Bulletin of Mathematical Biology}, 68(8):1923--1944, 11 2006.

\bibitem{SampleAllen2017}
Christine Sample and Benjamin Allen.
\newblock {The limits of weak selection and large population size in
  evolutionary game theory}.
\newblock {\em Journal of Mathematical Biology}, 75(5):1285--1317, 11 2017.

\bibitem{HauertDoebeli2004}
{Christoph Hauert} and {Michael Doebeli}.
\newblock {Spatial structure often inhibits the evolution of cooperationin the
  snowdrift game}.
\newblock {\em Nature}, 428:643--646, 4 2004.

\bibitem{DoebeliHauert2005}
Michael Doebeli and Christoph Hauert.
\newblock {Models of cooperation based on the Prisoner's Dilemma and the
  Snowdrift game}, 7 2005.

\bibitem{Skyrms2001}
Brian Skyrms.
\newblock {The Stag Hunt}.
\newblock {\em Proceedings and Addresses of the American Philosophical
  Association}, 75(2):31--41, 2001.

\bibitem{Wild2007WeakSelection}
Geoff Wild and Arne Traulsen.
\newblock {The different limits of weak selection and the evolutionary dynamics
  of finite populations}.
\newblock {\em Journal of Theoretical Biology}, 247(2):382--390, 7 2007.

\bibitem{Ohtsuki2006ANetworks}
Hisashi Ohtsuki, Christoph Hauert, Erez Lieberman, and Martin~A. Nowak.
\newblock {A simple rule for the evolution of cooperation on graphs and social
  networks}.
\newblock {\em Nature}, 441(7092):502--505, 5 2006.

\bibitem{Pattni2017Subpopulations}
Karan Pattni, Mark Broom, and Jan Rycht{\'{a}}{\v{r}}.
\newblock {Evolutionary dynamics and the evolution of multiplayer cooperation
  in a subdivided population}.
\newblock {\em Journal of Theoretical Biology}, 429:105--115, 9 2017.

\bibitem{Pattni2015EvolutionaryProcess}
Karan Pattni, Mark Broom, Jan Rycht{\'{a}}{\v{r}}, and Lara~J. Silvers.
\newblock {Evolutionary graph theory revisited: When is an evolutionary process
  equivalent to the Moran process?}
\newblock {\em Proceedings of the Royal Society A: Mathematical, Physical and
  Engineering Sciences}, 471(2182), 10 2015.

\bibitem{Karlin1975Stochastic}
Samuel Karlin and Howard~M. Taylor.
\newblock {\em {A First Course in Stochastic Processes}}.
\newblock Academic Press, New York, USA, second edition, 1975.

\bibitem{Traulsen2009Stochastic}
Arne Traulsen and Christoph Hauert.
\newblock {Stochastic evolutionary game dynamics}.
\newblock {\em Reviews of Nonlinear Dynamics and Complexity}, 2:25--61, 11
  2009.

\bibitem{Broom2012Framework}
Mark Broom and Jan Rycht{\'{a}}{\v{r}}.
\newblock {A general framework for analysing multiplayer games in networks
  using territorial interactions as a case study}.
\newblock {\em Journal of Theoretical Biology}, 302:70--80, 6 2012.

\bibitem{Broom2015AInteractions}
Mark Broom, Charlotte Lafaye, Karan Pattni, and Jan Rycht{\'{a}}{\v{r}}.
\newblock {A study of the dynamics of multi-player games on small networks
  using territorial interactions}.
\newblock {\em Journal of Mathematical Biology}, 71(6-7):1551--1574, 3 2015.

\bibitem{Overton2022QuasiStationarySIS}
Christopher~E. Overton, Robert~R. Wilkinson, Adedapo Loyinmi, Joel~C. Miller,
  and Kieran~J. Sharkey.
\newblock {Approximating Quasi-Stationary Behaviour in Network-Based SIS
  Dynamics}.
\newblock {\em Bulletin of Mathematical Biology}, 84(1), 1 2022.

\bibitem{Hamilton1967}
William~D. Hamilton.
\newblock {Extraordinary Sex Ratios}.
\newblock {\em Science}, 156:477--488, 1967.

\bibitem{Wu2013MultiplayerGames}
Bin Wu, Arne Traulsen, and Chaitanya~S. Gokhale.
\newblock {Dynamic properties of evolutionary multi-player games in finite
  populations}.
\newblock {\em Games}, 4(2):182--199, 6 2013.

\bibitem{Broom2019SocialDilemmas}
Mark Broom, Karan Pattni, and Jan Rycht{\'{a}}{\v{r}}.
\newblock {Generalized Social Dilemmas: The Evolution of Cooperation in
  Populations with Variable Group Size}.
\newblock {\em Bulletin of Mathematical Biology}, 81(11):4643--4674, 11 2019.

\end{thebibliography}

%\end{multicols}

\end{document}